\documentclass[12pt]{article}

\usepackage{newtxtext,newtxmath}

\usepackage{graphicx}
\usepackage{algorithm}
\usepackage{algpseudocode}
\usepackage{amsmath}
\usepackage{verbatim}

\usepackage[letterpaper,margin=1in]{geometry}

\renewenvironment{abstract}
	{\quotation}
	{\endquotation}

\date{}

\makeatletter
\renewcommand{\fnum@figure}{\textbf{Figure \thefigure}}
\renewcommand{\fnum@table}{\textbf{Table \thetable}}
\makeatother

\usepackage{scicite}

\usepackage{url}

\def\scititle{
	A Few Reveal Many: Constraint-Enhanced Learning of Many-Body Quantum Systems
}
\title{\bfseries \boldmath \scititle}

\author{
	Matteo Tedde$^{1}$,
	Davide Girolami$^{1\ast}$\and
	\small$^{1}$Department of Applied Science and Technology, Politecnico di Torino, Torino, Italy.\and
	\small$^\ast$Corresponding author. Email: davegirolami@gmail.com.
}

\begin{document} 

\maketitle

\begin{abstract} \bfseries \boldmath
Discovering properties of many-body quantum systems is challenging due to the large number of   parameters involved. Here, we show that  harnessing fundamental constraints among quantum observables improves  quantum learning,  enabling the inference of exponentially many quantities from a linear number of direct measurements. 
We leverage this result to design constraint-enhanced algorithms for evaluating  key quantum features such as Entanglement
and Magic in 100-qubit simulators. At a fixed number of  measurement shots, they achieve up to an 83\% reduction in   relative error compared with    protocols based on constructing ``classical shadows" of quantum states. 
Quantum constraints themselves are therefore a resource for exploring large
quantum systems, guiding the allocation of  the measurement budget to highly informative observables.
\end{abstract}

\noindent Many-body quantum systems are the building blocks of matter and the hardware of promising quantum technologies~\cite{Preskill2018}. Yet, certifying their features is hard,  because they depend on a  number of parameters that grows exponentially with their components. Further,  direct measurements alter quantum states. Hence, every shot  requires preparing  an independent replica.

\noindent Here, we show that fundamental constraints on measurement outcomes  are a resource for quantum learning. 
In particular, we  prove that the expectation values of exponentially many quantities can be inferred with guaranteed precision, by probing   a   set of observables increasing linearly with the number of system components. 
 \\
\noindent   Building on this theoretical result, we  introduce ``constraint-enhanced" (ConEnh) learning algorithms.  They 
concentrate an initially unstructured exploration of the measurement landscape  to directly observe quantities that are still highly uncertain, i.e., weakly constrained by previously collected data. That is, information from quantum constraints guides the derandomization of the measurement budget.\\
First, we perform ConEnh learning  of  Entanglement, the most important form of quantum correlations \cite{RevModPhys.81.865}. We simulate the preparation of 100-qubit two-dimensional (2D) cluster states \cite{DurVidalCirac2000,RaussendorfBriegel2001}, which are then perturbed  by applying local  rotations. This choice is motivated by the importance of these states for quantum computing and by their complex correlation structure. The ConEnh learning algorithm computes the bipartite  2-Rényi entanglement entropy of these states \cite{Islam2015,Brydges2019},   extracting information from constraint-selected Pauli measurements, with no premeasurement entangling operations allowed. We demonstrate up to a 43\% reduction in relative error at fixed measurement shots against  a classical-shadow algorithm (CS), a recently proposed method to predict quantum properties from randomized measurements \cite{Aaronson2018,Huang2020,Zhang2021,Struchalin2021,Chen2021,KohGrewal2022,Zhang2024Quantification}.\\
In a second series of tests, we exploit quantum constraints to boost evaluation of quantum state Magic (also known as ``nonstabilizerness"), which identifies quantum states that provide a computational advantage against classical devices~\cite{BravyiKitaev2005,Veitch2012,Howard2014,SeddonCampbell2019}. 
 We consider  100-qubit 2D cluster states  and GHZ states, another important class of entangled states ~\cite{greenberger1989going,DurVidalCirac2000,RaussendorfBriegel2001}, which are then injected with Magic by applying local  rotations. We evaluate the performance of two algorithms that compute the stabilizer 2-Rényi entropy of the final states \cite{RevModPhys.81.865,Leone2022,Oliviero2022MeasuringMagic}. The task is more demanding than Entanglement detection, as it requires measuring the global system. To overcome this difficulty, we construct the neural  approximation of quantum states via the Autoregressive Gram-Hadamard Density Operator (AGHDO) ansatz~\cite{CarleoTroyer2017,Torlai2018,TorlaiMelko2018,Carrasquilla2019,VicentiniRossiCarleo2022,PhysRevLett.124.020503}.  
 We then develop a constraint-enhanced AGHDO (ConEnh-AGHDO)  protocol for quantifying Magic  and compare it against a classical shadows competitor (CS-AGHDO). The results show that ConEnh-AGHDO achieves up to an $83\%$ reduction in relative error at fixed shots, demonstrating how constraint-informed measurement strategies enhance  machine learning.
 Further, the difference between ConEnh and CS performance increases with loss of information about the system. 

\section*{Results}
\subsection*{\label{sec:the-results}Information propagation among quantum observables}
Consider a state $\rho$ of $n$ qubits. A polynomial function $f(\rho)$ can be written in terms of the expectation values $\mu_P=\operatorname{Tr}(\rho P)$ of  Pauli observables $P\in \mathcal P_n=\{I,X,Y,Z\}^{\otimes n}$: 
\begin{equation}
    f(\rho)=\sum_{\alpha}c_\alpha \prod_{j=1}^{d_\alpha}\mu_{P_{\alpha_j}},\label{eq:polynomial-property}
\end{equation}
where $c_\alpha$ are real coefficients that weigh each term of degree $d_\alpha$. \\ 
Since every Pauli pair either commutes or anticommutes, direct estimation of a few   $\mu_P$ sets bounds on unmeasured expectation values across 
${\cal P}_n$ (Fig.~\ref{fig1}), and therefore on the target function. For instance, a single qubit measurement $\mu_X  \approx 1$ determines $\mu_{Y,Z}\approx 0$. Also,  joint measurability of commuting observables $P,Q$ implies   that  if  one measures $ \mu_P\approx 1$ then  $ \mu_Q\approx \mu_{ PQ}$.\\We quantify how information from  a given set of  finite-shot measurements propagates in ${\cal P}_n$. We find 
 that probing a set of observables linearly increasing with the system size is sufficient to infer with guaranteed precision exponentially many unmeasured elements of ${\cal P}_n$:\\
{\bf Theoretical Result --} 
\label{thm:pauli-propagation}
Let $G_1,\dots,G_r, 1\le r\le n,$ be independent commuting Pauli observables. For each $i$, measure $G_i$ on $M$ independent copies of the state, obtaining sample means $\hat \mu_{G_i}$. Call $H_i=\text{sgn}(\hat \mu_{G_i})G_i$ and for $0<\delta<1$, define
\begin{equation}
a_i:=\sqrt{\frac{2\ln (2r/\delta)}{M}}, \qquad e_i:=\min\{2,1-|\hat \mu_{G_i}|+a_i\},\qquad L_i:=(|\hat \mu_{G_i}|-a_i)_+.
\end{equation}
Then, with probability at least $1-\delta$, the following statements hold:
\begin{enumerate}
    \item $\mu_{H_i}\ge 1-e_i$, for every $i=1,\dots,r$.
    \item If an unmeasured observable $Q$ anticommutes with at least one $G_i$, then
    \begin{equation}
        |\mu_Q|\le \min_{i:\{Q,G_i\}=0}\sqrt{1-L_i^2}.
        \label{eq:anticommutingbound}
    \end{equation}
    \item For every $K\subseteq \{1,\dots,r\}$, let $H_K=\prod_{i\in K}H_i$. Then
    \begin{equation}
        \mu_{H_K}\ge \max\bigg\{-1,1-\sum_{i\in K}e_i\bigg\}.
    \end{equation}
\end{enumerate}
The proof is reported in Methods. \\
These inequalities  characterize indirect inference of quantum observables. They incorporate  theoretical bounds dictated by  quantum theory with 
confidence intervals defined by the  Hoeffding radii $a_i$\cite{BoucheronLugosiMassart2013,Hoeffding1963}.  
The fraction of  observables left unconstrained in ${\cal P}_n$ is $(2^{2n-r}-2^r)/4^n\leq 2^{-r}$, which  is exponentially suppressed by linearly increasing $r$. 
Directly probing commuting  Pauli observables is also experimentally appealing, as it reduces the number of measurement circuits \cite{PhysRevX.10.031064}. Yet, any rotation globally applied to ${\cal P}_n$  would generate a set of parameters subject to the same mutual constraints. \\
Noticeably, 
when nearly deterministic observables are measured, $1-|\hat \mu_{G_i}|=o(a_i)$, 
the upper bounds on   anticommuting Pauli strings scale as $O((\ln (2r/\delta)/M)^{1/4})$. Products of the measured generators are also constrained, with errors controlled by the sum of the individual confidence radii.  \\
Identifying these most  informative   observables is possible with prior information about the state, e.g., it is a stabilizer state with known $r$ generators \cite{Gottesman1997}. That is, one can verify to have prepared such states in a laboratory with a linear number of measurements, with guaranteed precision  dictated by our result. Further,  our tests will show that even if the information about the quantum state is partially lost, quantum constraints still prove to be a computational resource.

\subsection*{\label{sec:comp-results}Constraint-enhanced learning algorithms}
Our theoretical result bounds the contributions of unmeasured Pauli expectation values to target quantum properties, enabling physically informed  estimation. Also, it identifies observables that remain almost unconstrained by previously acquired information, providing a criterion for adaptive measurement allocation.\\
We exploit both mechanisms,   
 by developing a ConEnh algorithm  to learn the  bipartite  2-Rényi entanglement entropy of pure states $\psi_{AB}$\cite{RevModPhys.81.865}:
\begin{equation}{\cal E}_2(\psi_{AB}):=-\log_2 \text{Tr}(\rho^2_A),\qquad \text{Tr}\left(\rho_A^2\right) = 2^{-|A|}\sum_{P_A\in \mathcal P_{|A|}}\mu_{P_A}^2, 
\end{equation}
where the marginal purity of a subsystem $A$ is expressed as a function of Pauli observables $P_A$ of $A$. The quantity meets several desirable properties of entanglement measures: it takes zero value only for product states (faithfulness), it is invariant under local unitary transformations, and it cannot increase under deterministic local operations and classical communication \cite{10.1063/1.3511477}.\\
The pseudocode of the ConEnh  algorithm is illustrated in Table 1, while full details are given in Supplementary Material.  \\
We test it to learn Entanglement of 100-qubit states. Specifically, we consider   2D cluster states subject to   local rotations, which we construct as follows.
Let $G=(V,E)$ be a rectangular grid with $n=|V|$ and 
\begin{equation}
    |C_G\rangle = \prod_{(u,v)\in E}CZ_{uv}|+\rangle ^{\otimes n}
    \label{eq:cluster-state}
\end{equation}
be the corresponding cluster state~\cite{BriegelRaussendorf2001,RaussendorfBriegel2001,RaussendorfBrowneBriegel2003,Hein2004} with generators
    $H_i=X_i \prod_{j\in N(i)}Z_j,\, i\in V.
    $ 
   Applying $k$ single-qubit unitaries, one obtains 
\begin{equation}
    |\psi^{\text{2D}} _{R,\theta}\rangle = \prod_{j \in R}e^{i\theta_j Z_j}|C_G\rangle, \qquad |R|=k.
    \label{eq:rotated-cluster}
\end{equation}
Although the rotations drive some generator expectation values away from their saturated values, they remain informative because of the constraints set by our Theoretical Result. 
The Rényi-2 entanglement entropy  remains unchanged for any bipartition under local rotations \cite{hein2006entanglementgraphstatesapplications}:
\begin{equation}
{\cal E}_2(
 |\psi^{\text{2D}} _{R,\theta}\rangle)=|A|-\log_2|\mathcal G_A|,\end{equation}
where $|\mathcal G_A|$ is the number of stabilizer group elements supported within $A$.\\
The simulation shows that ConEnh learning reduces the relative error at fixed shots up to $43\%$ with respect to the standard CS protocol (Fig.~\ref{figent}) in the most complex studied scenario of $|A|=50$. 
The latter relies on  randomized Pauli measurements that generate classical snapshots (shadows) of quantum  states \cite{Aaronson2018}. The target function  is then computed from the shadow state \cite{Brydges2019,Huang2020,Oliviero2022MeasuringMagic,Zhang2024Quantification}.
Both protocols bypass the experimental challenge  to simultaneously control multiple replicas of quantum states, as required by interferometric and SWAP entanglement detection schemes\cite{Ekert2002,Islam2015,Kaufman2016}.\\ In this test, we assume prior knowledge of the stabilizer generators of the input state, while the locations and angles of the subsequently applied rotations are unknown to the learner. ConEnh uses the reference generators to guide measurement selection through quantum constraints, while 
CS constructs its measurement schedule from the same reference generators supplied to ConEnh but keeps it fixed throughout acquisition. 
Under this assumption, both  ConEnh and CS algorithms interrogate the system by direct Pauli measurements, without the possibility to use  premeasurement entangling gates or additional Clifford or non-Clifford rotations \cite{Huang2020}, which are expensive computational resources.  
Specifically,  given local-Pauli measurement  $B=(B_1,\dots,B_n)$, where $B_j\in \{X,Y,Z\}$, CS samples local bases covering the reference generators, mixed with uniform local-Pauli exploration using the same exploration weight as ConEnh. Its bases and their probabilities are fixed before data collection and are not updated from measurement outcomes. The ConEnh method replaces this fixed schedule with a distribution $q_t(B)$ determined before observing the outcome at time $t$ by employing quantum constraints to select direct Pauli measurements. 
For direct estimation,  ConEnh and CS share the moment estimator, the Pauli-sampling pool when used, and a ``projection" rule. Their comparison therefore evaluates adaptive constraint-guided acquisition, also boosted by  constraint-informed post-processing (``clipping"), against a specified static design with the same prior information. The construction and exact coverage probabilities are given in the Supplementary Material.\\


\noindent In a second test, we design ConEnh learning of 100-qubit pure state Magic, which certifies deviation from the class of classically simulable  stabilizer states \cite{Gottesman1997,Howard2014}. The target function is the stabilizer 2-Rényi entropy~\cite{Leone2022,LiuWinter2022}:
\begin{align}\label{eq:stabilizer-Rényi}
    \mathcal M_2(\psi) &:= -\log_2\mathcal A_4(\psi),\qquad \mathcal A_4(\psi)=2^{-n}\sum_{P\in \mathcal P_n}\mu_P^4.
\end{align}
  This measure satisfies the  criteria of faithfulness, invariance under Clifford rotations, and monotonicity under deterministic stabilizer protocols. Since it does not need optimization over stabilizer decompositions, it is both computationally and experimentally appealing, as demonstrated in a 7-qubit quantum computer \cite{Oliviero2022MeasuringMagic}.\\
The stabilizer 2-Rényi entropy of the rotated cluster states in Eq.~(\ref{eq:rotated-cluster}) is
\begin{equation}
    \mathcal M_2(|\psi^{\text{2D}} _{R,\theta}\rangle )=-\sum_{j\in R}\log_2\bigg[\frac{1+\cos ^4(2\theta_j)+\sin^4(2\theta_j)}{2}\bigg].
    \label{eq:rotated-cluster-magic}
\end{equation}
As a second class of states to scrutinize, we select the rotated GHZ states
\begin{equation}\label{ghz}
|\psi^{\text{GHZ}}_{R,\theta}\rangle=e^{i\theta_0Z_0}\prod_{j\in R}e^{i\theta_jX_j}|\psi^{\text{GHZ}}\rangle, \,\,|\psi^{\text{GHZ}}\rangle=\frac{|0\rangle^{\otimes n}+|1\rangle^{\otimes n}}{\sqrt{2}}.
\end{equation}
In this case, the stabilizer 2-Rényi entropy  is
\begin{equation}
    \mathcal M_2(|\psi^{\text{GHZ}} _{R,\theta}\rangle )=-\sum_{j\in \{0\}\cup R}\log_2\bigg[\frac{1+\cos ^4(2\theta_j)+\sin^4(2\theta_j)}{2}\bigg].
    \label{eq:rotated-cluster-magic}
\end{equation}
Magic estimation is enhanced  by embedding ConEnh methods into the AGHDO neural state approximation algorithm 
~\cite{Hradil1997,TorlaiMelko2018,Carrasquilla2019,VicentiniRossiCarleo2022}. 
We represent the state by
\begin{equation}\rho_\phi=C\bigg(\bigotimes_{i=1}^n\varrho_{\phi,i}\bigg)C^\dagger,\qquad \varrho_{\phi,i}(s,t)=\sum_{a=1}^R\psi_{\phi,i}(s,a)\psi^*_{\phi,i}(t,a),
    \label{eq:aghdo-density}
\end{equation}
where $s,t\in\{0,1\}$, $C$ is the reference Clifford transformation associated with the GHZ and  2D cluster  states, and $\psi_{\phi,i}$ is parametrized by a neural network through $\phi$. The neural amplitudes satisfy
\begin{equation}
    \sum_{s=0}^1\sum_{a=1}^R|\psi_{\phi,i}(s,a)|^2=1.
\end{equation}
Crucially, our theoretical constraints are applicable to 
the neural state in Eq.~\eqref{eq:aghdo-density}, as it is a physical density matrix. 
We can then construct ConEnh-AGHDO (pseudocode in Table~\ref{tab:placeholder}) and CS-AGHDO algorithms that share the same neural parametrization, likelihood, optimizer and hyperparameters. Each model is fitted to its own measurement records. ConEnh-AGHDO updates its schedule using the observed records, while CS-AGHDO uses the fixed stabilizer-informed acquisition schedule.\\
We assume that the learner knows the reference Clifford circuit $C$ and the product structure in $\rho_\phi$, but  
not what qubits are rotated and about what angles. The single-mode density operators are inferred from  Pauli measurements.
ConEnh complements the reconstruction by identifying measurement directions that our bounds indicate as informative. For learning Magic, we use the factorized logical-mode version described in Supplementary Material. Pauli-space sampling offers another route for structured states~\cite{LamiCollura2023}.\\
As shown in Fig.~\ref{figmagic}, ConEnh-AGHDO reduces the relative error at fixed shots up to 83\% for the cluster state and $k=50$ rotated qubits, while the discrepancy in performance between the two protocols grows by increasing  $k$. That is, the ConEnh advantage increases with loss of information about the systems. Since the estimator and model class are identical, this difference is strong evidence of an advantage provided by the ConEnh acquisition policy. Note that, unlike direct estimators, AGHDO also introduces model error through its parametrization of the physical family, as detailed in the Supplementary Material.

\section*{\label{sec:discussions}Discussion}
We proposed and tested a physical principle for adaptive quantum learning. Exploiting fundamental constraints among quantum observables enables extracting information about exponentially many quantities from a pool of direct measurements scaling linearly with the system components. Embedding this additional  knowledge into canonical protocols enhances learning of Entanglement and Magic.\\ 
 ConEnh learning does not require simultaneous control of many copies of the system ~\cite{Ekert2002,Islam2015}, making it testable for near-term quantum devices. 
We therefore expect it to be employable in a broad range of quantum-information tasks. These include the evaluation of other quantum state  polynomials, Chebyshev/Taylor   polynomial approximations of generic functions,    such as the von Neumann entropy and related entanglement quantifiers,  learning of Hamiltonians and their ground states \cite{Huang2020,bla},   and   full state reconstruction~\cite{Hradil1997,James2001,Haah2017}. Further, it provides a promising  derandomization strategy for scenarios in which the complexity of the studied system or the target function is controllable \cite{Gross2010,Cramer2010,Flammia2012,Baumgratz2013,Lanyon2017, Huang2021Derandomization,Hadfield2022}.   \\
More generally, we anticipate  that feeding information from fundamental quantum constraints into learning protocols offers a promising route towards exploration, control, and benchmarking of quantum devices.

\newpage
\begin{figure}[t!]
\centering
\includegraphics[width=.9\textwidth,height=8cm]{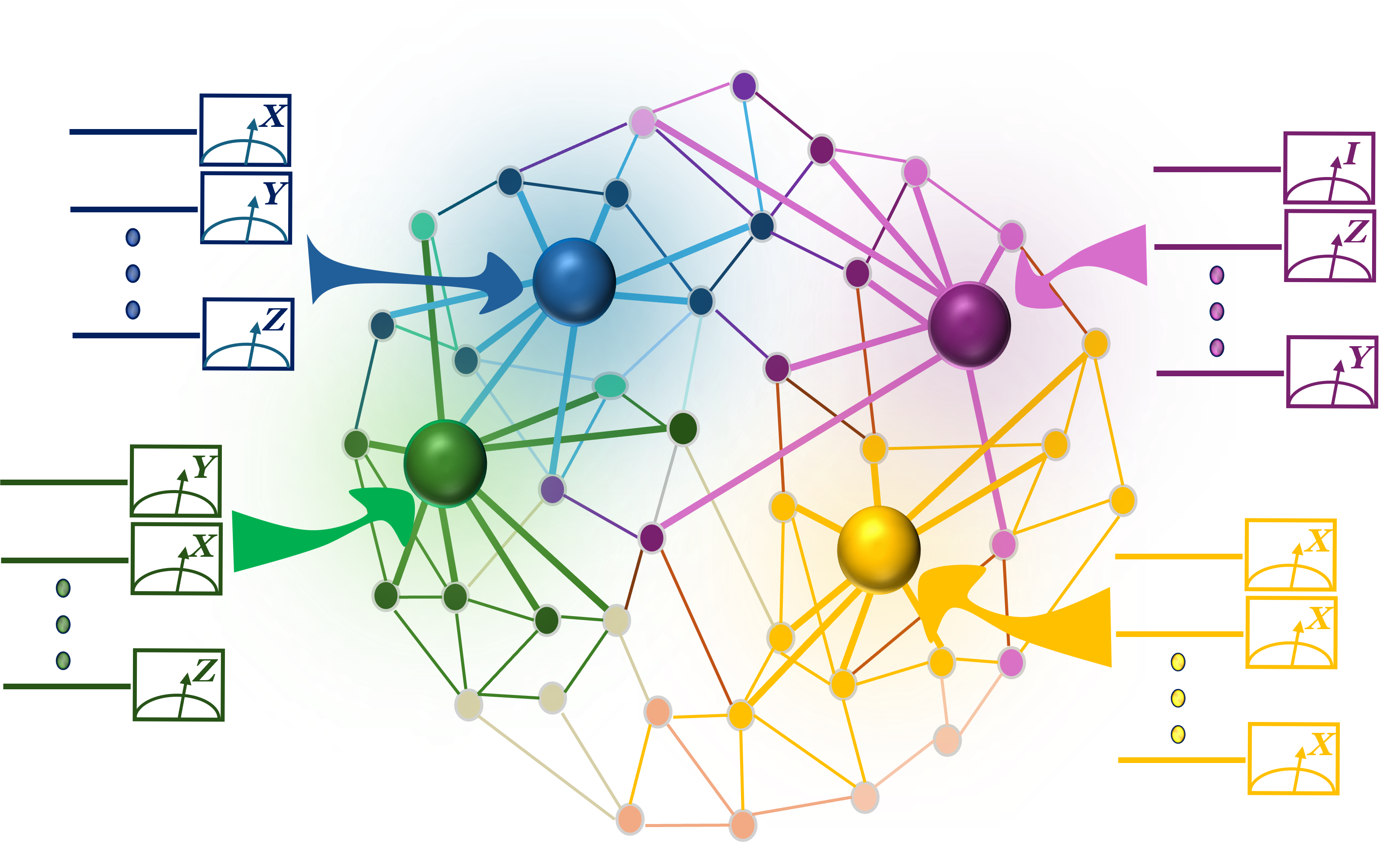}
\caption{
{\bf A few measurements reveal many observables}.
Measuring a few  Pauli observables (here depicted as the biggest vertices in a connected network)--in fact, even a single Pauli coefficient--carries information about an exponential number of unobserved Pauli strings (the rest of the network). We quantify the  strength of such bounds and their usefulness for quantum learning. Extremal expectation values ($\pm1$) of commuting observables  yield tight limits to unmeasured quantities.  A priori information about the scrutinized system, e.g., spatial symmetries or justified assumptions, can also inform the initial  selection of the measured observables, which can  then be adapted by updating the quantum constraints during a learning protocol.}
\label{fig1}
\end{figure}

\begin{figure}
    \hspace{-0pt}\includegraphics[width=1\columnwidth]{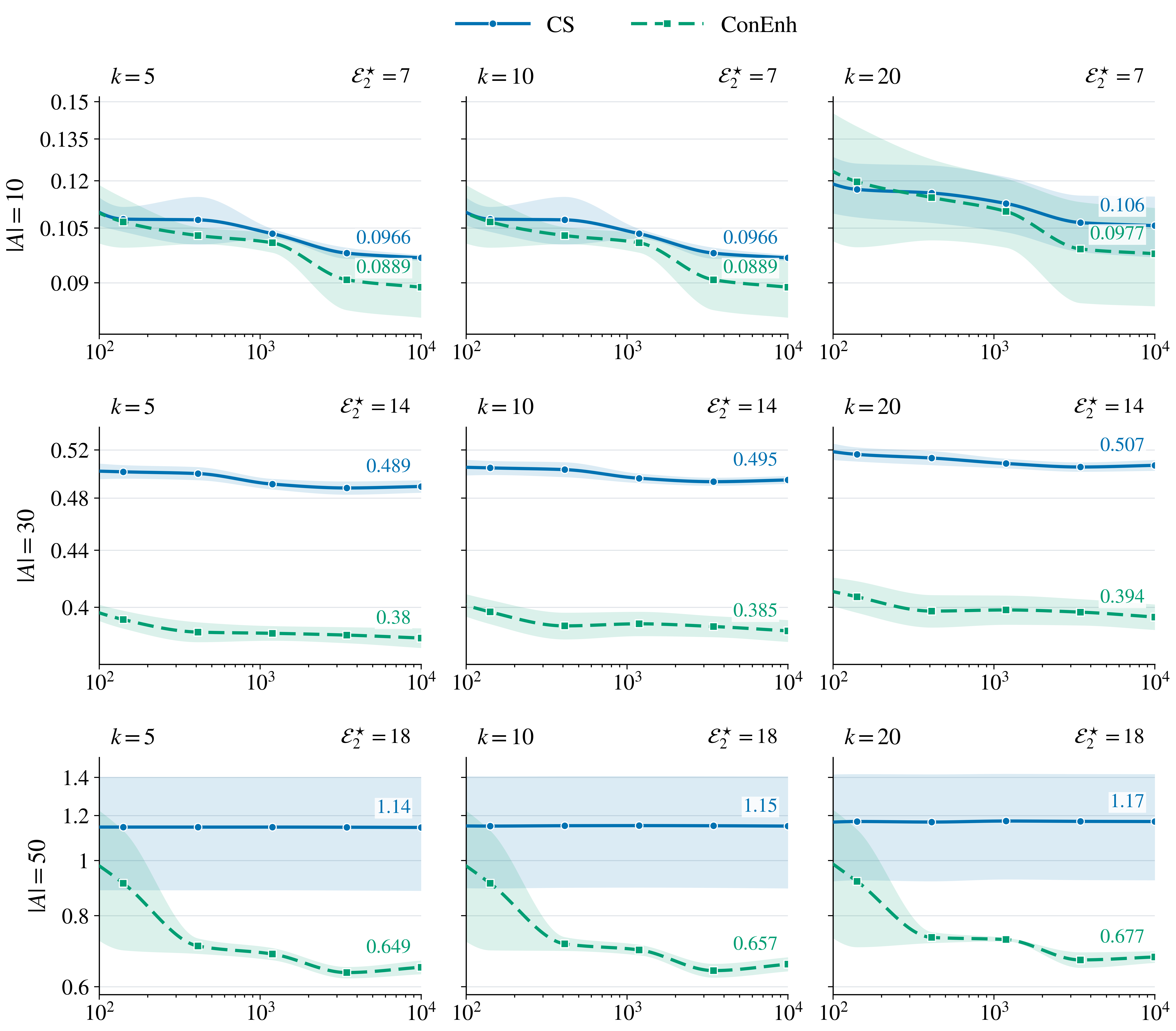}
    \caption{{\bf ConEnh vs CS Entanglement learning.} We compute  the 2-Rényi  entanglement entropy in 100-qubit rotated 2D cluster states (Eq.~\ref{eq:rotated-cluster})  between $|A|=10,30,50$  connected qubits and the rest of the system. 
    Given the theoretical values $\mathcal E_2^\star$, the plots track the relative error $|\mathcal E_2^\star-\hat{\mathcal E_2}|/\mathcal E_2^\star$  (vertical axes) at varying measurement shots (horizontal axes) of the ConEnh algorithm (dashed green line, pseudocode in Table~\ref{tab1})  against a CS-based protocol (blue line). The lowest values, reached at $10^4$ shots, are written explicitly.
We vary the number of rotated qubits, $k=5,10,20$, and the rotation magnitude $\theta_1,\dots,\theta_k\overset{\mathrm{i.i.d.}}{\sim} \text{Uniform}(\pi/16,\pi/8)$.  Each plotted point is the mean of 5 independent repetitions. The shaded bands show the standard error of the mean.\\
ConEnh advantage increases with $|A|$ while being stable as $k$ increases, reducing the error up to $43\%$ for $|A|=50, k=5$.}
    \label{figent}
\end{figure}

 \begin{figure}[htbp]
    \hspace{-0pt}
    \includegraphics[width=\columnwidth,height=10cm]{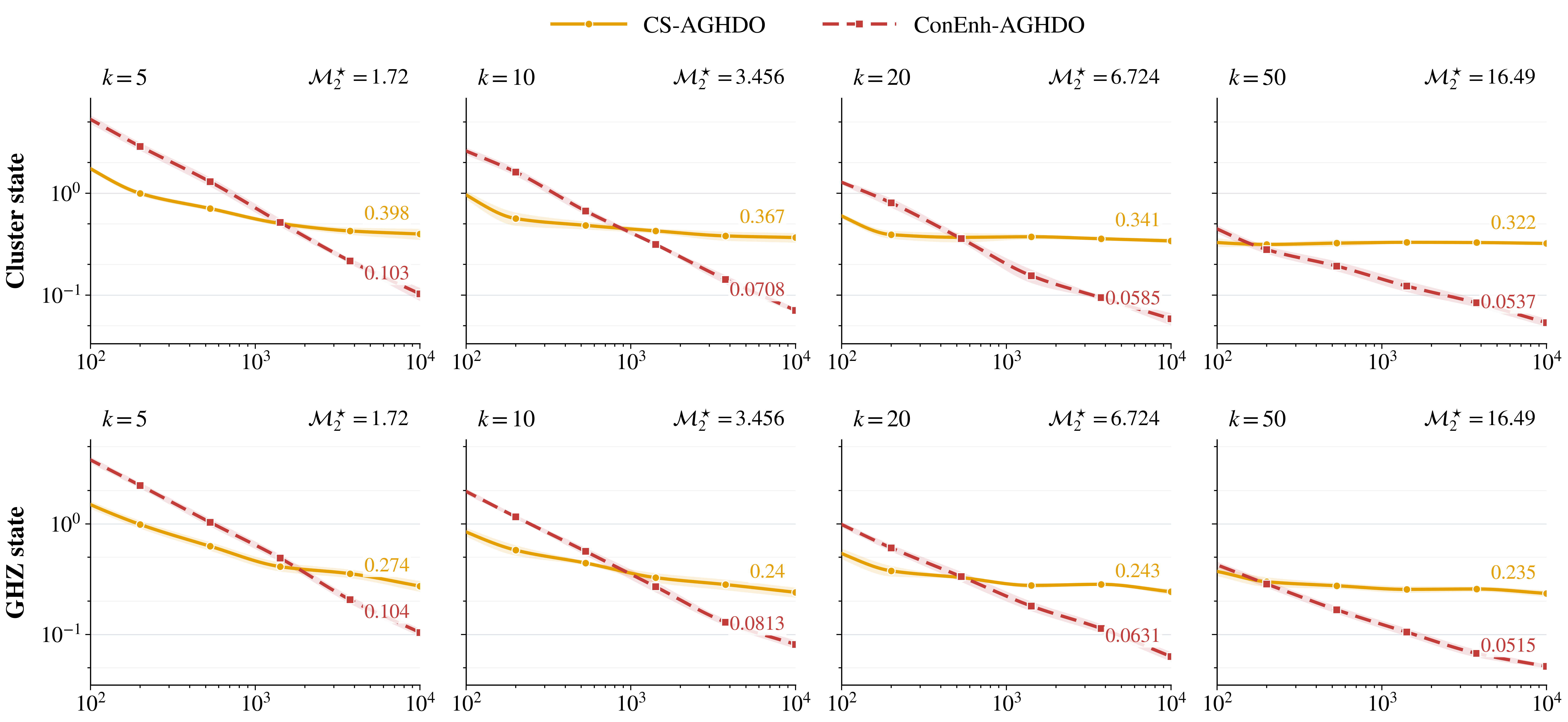}\caption{{\bf ConEnh-AGHDO vs CS-AGHDO Magic learning.}  We compute  the stabilizer 2-Rényi  entropy in 100-qubit rotated 2D cluster and GHZ states (Eq.~(\ref{eq:rotated-cluster},\ref{ghz})). 
    The plots depict the relative error ($|\mathcal M_2^\star-\hat{\mathcal M_2}|/\mathcal M_2^\star$) (vertical axes) at varying shots (horizontal), of the ConEnh-AGHDO algorithm (dashed red line, pseudocode in Table~\ref{tab:placeholder})  against a CS-AGHDO protocol (orange line).
We vary the number of rotated qubits  $k=5,10,20,50$, and the rotation magnitude  $\theta_1,\dots,\theta_k\overset{\mathrm{i.i.d.}}{\sim} \text{Uniform}(\pi/16,\pi/8)$.   Each plotted point is the mean of 5 independent repetitions, and the shaded bands (much smaller than those of Fig.(\ref{figent})) show the standard error of the mean.\\
 CS performs better with fewer shots, as it concentrates the measurements on the pre-rotation input state stabilizers, while ConEnh has still to identify the most informative Pauli observables. The line crossover occurs at smaller shot numbers by increasing $k$, as the information about the input state becomes less relevant. The reduction in relative error at $10^4$ shots by ConEnh peaks at $83\%$ for the cluster state with $50$ rotated qubits.}
    \label{figmagic}
\end{figure}

\clearpage 

%
\begin{table}[H]
\vspace{-20pt}
\begin{algorithm}[H]
{\small
\caption{ConEnh for Entanglement}
\label{alg:conenh_entanglement}
\begin{algorithmic}[1]

\Require Copies of $\rho$; subsystem $A$;
known generators $\mathcal H=\{H_i\}$;
shots $N$; Pauli pool size $S$
\Ensure Estimate $\hat{\mathcal{E}}_2(A)$

\Statex \textbf{Internal variables:}
\Statex \hspace{\algorithmicindent}
$\Gamma_t$: generator, coverage, and local-axis
statistics after shot $t$; $\mathcal U$: set of distinct sampled Pauli operators; $F_P,T_P$: accumulators for the second-moment estimator; $\mathcal C_t$: set of candidate measurement settings; $\mathcal S_t(c)$: score assigned to candidate $c$

\State Draw $S$ nonidentity Pauli operators
$P\sim\nu_A(P)$ and construct $\mathcal U$
\State Assign $W_P\gets c_P/(S\nu_A(P))$
for each $P\in\mathcal U$
\State Initialize $\Gamma_0$ and
$F_P=T_P=0$ for all $P\in\mathcal U$

\For{$t=1,\ldots,N$}

    \If{a policy refresh is due}
        \State $(\mathcal C_t,\mathcal S_t)
        \gets\Call{DesignCandidates}
        {A,\mathcal U,\mathcal H,\Gamma_{t-1}}$ (See details in Supplementary Text)
    \EndIf

    \State $p_t(c)\gets
    \frac{\exp(\mathcal S_t(c)/\tau)}
    {\sum_{c'\in\mathcal C_t}
    \exp(\mathcal S_t(c')/\tau)}$

    \State $q_t(B)\gets
    (1-\eta)\sum_{c\in\mathcal C_t}
    p_t(c)\mathbf 1[B=B^{(c)}]
    +\eta3^{-n}$

    \State $\pi_t(P)\gets
    \sum_{B\succeq P}q_t(B)$
    for all $P\in\mathcal U$

    \State Draw $B_t\sim q_t$ and measure
    a new copy of $\rho$
    \State Store the outcomes $\mathbf b_t$

    \For{each $P\in\mathcal U$}
        \State $Z_t(P)\gets
        \frac{\mathbf 1[B_t\succeq P]b_t(P)}
        {\pi_t(P)}$
        \State $T_P\gets T_P+F_PZ_t(P)$; $F_P\gets F_P+Z_t(P)$
    \EndFor

    \State Update $\Gamma_t$ using all covered
    generator parities and measurement statistics

\EndFor

\State $U_N(P)\gets
\frac{T_P}{\binom{N}{2}}$
for all $P\in\mathcal U$

\State Compute empirical generator means
$\hat h_i$ from $\Gamma_N$, setting
$\hat h_i=0$ if unobserved

\State $v_P\gets
\min\left(\{1\}\cup \{
1-\hat \mu_{H_i}^2:\{P,H_i\}=0\}
\right)$

\State $\hat p_A\gets
2^{-|A|}
\left[
1+\sum_{P\in\mathcal U}
W_P\operatorname{clip}(U_N(P),0,v_P)
\right]$

\State $\tilde p_A\gets
\operatorname{clip}
(\hat p_A,2^{-|A|},1)$

\State \Return
$\hat{\mathcal E}_2(A)\gets-\log_2\tilde p_A$

\end{algorithmic}
}
\end{algorithm}
\vspace{-30pt}
    \caption{{\bf Pseudocode of ConEnh algorithm}}
    \label{tab1}
\end{table}

\newpage

\begin{table}
\vspace{-20pt}
\begin{algorithm}[H]
{\small
\caption{ConEnh-AGHDO for Magic}
\label{alg:conenh_aghdo_magic}
\begin{algorithmic}[1]
 
\Require Copies of $\rho$;
known Clifford transformation $C$;
known generators $\mathcal H=\{H_i\}$;
shot budget $N$; checkpoints $\mathcal T$
with $N\in\mathcal T$
\Ensure Estimate $\hat{\mathcal M}_2$

\State Define the focus observables
$Q_{i,\alpha}=C\sigma_i^\alpha C^\dagger$,
$\alpha\in\{X,Y,Z\}$

\State Initialize $\Gamma_0$, the outcome
counts $N_{i,\alpha}^{\pm}=0$,
and neural parameters $\phi$

\State Parametrize the state as
$\rho_\phi=
C\left(\bigotimes_{i=0}^{n-1}
\rho_{i,\phi}\right)C^\dagger$

\For{$t=1,\ldots,N$}

    \State Compute empirical generator means
    $\hat \mu_{H_i}$ and counts $n_i$ from $\Gamma_{t-1}$

    \State $\hat e_i\gets
    \min\{1,1-|\hat \mu_{H_i}|+(1+n_i)^{-1/2}\}$

    \State $w_i\gets
    [\hat e_i(2-\hat e_i)]^2$

    \State $r_{i,\alpha}\gets
    w_i/(1+m_{i,\alpha})$

    \State Build 64 full Pauli-basis candidates
    $B^{(j)}$ by greedily packing compatible
    focus observables

    \State $\mathcal S_t(j)\gets
    \frac{
    \sum_{i,\alpha:
    B^{(j)}\succeq Q_{i,\alpha}}r_{i,\alpha}
    }{
    \sum_{i,\alpha}r_{i,\alpha}
    }$

    \State $p_t(j)\gets
    \frac{\exp(\mathcal S_t(j)/\tau)}
    {\sum_l\exp(\mathcal S_t(l)/\tau)}$

    \State $q_t(B)\gets
    (1-\eta)\sum_jp_t(j)
    \mathbf 1[B=B^{(j)}]
    +\eta3^{-n}$

    \State Draw $B_t\sim q_t$ and measure
    a new copy of $\rho$

    \State Extract all covered focus parities,
    including their Pauli signs

    \State Update outcome counts
    $N_{i,\alpha}^{\pm}$ and statistics $\Gamma_t$

    \If{$t\in\mathcal T$}

        \For{$l=1,\ldots,12$}

            \State Compute
            $\xi_{i,\alpha}=
            \operatorname{Tr}
            (\rho_{i,\phi}\sigma^\alpha)$

            \State $\mathcal L(\phi)\gets
            -\frac{1}{N_{\mathrm{rec}}}
            \sum_{i,\alpha}\sum_{b=\pm1}
            N_{i,\alpha}^{b}
            \log\left[
            \max\left\{
            10^{-12},
            \frac{1+b\xi_{i,\alpha}}{2}
            \right\}\right]$

            \State Update $\phi$ with Adam
            using $\nabla_\phi\mathcal L$

        \EndFor

        \State Compute the local density matrices
        $\rho_{i,\phi}$

        \State $\hat x_i\gets
        \operatorname{Tr}(\rho_{i,\phi}X)$; $\hat y_i\gets
        \operatorname{Tr}(\rho_{i,\phi}Y)$; $\hat z_i\gets
        \operatorname{Tr}(\rho_{i,\phi}Z)$

        \State $\hat M_2(t)\gets
        -\sum_{i=0}^{n-1}
        \log_2\left(
        \frac{1+\hat x_i^4+\hat y_i^4+\hat z_i^4}{2}
        \right)$

        \State $\hat{\mathcal M}_2(t)\gets
        \operatorname{clip}(\hat{\mathcal M}_2(t),0,n)$

    \EndIf

\EndFor

\State \Return $\hat{\mathcal M}_2(N)$
\end{algorithmic}
}
\end{algorithm}

\vspace{-30pt}
    \caption{{\bf Pseudocode of ConEnh-AGHDO algorithm}}
    \label{tab:placeholder}
    
\end{table}
\clearpage
\subsection*{Methods}

 \subsubsection*{Proof of Theoretical Result}
 \noindent{\it Constraints among anticommuting observables}\\
We recall the proof of bounds to anticommuting observables \cite{10.1063/1.2943685,Robertson1929,MacconePati2014,Guhne2004}.
Let $P_1,\dots,P_\ell$ be pairwise anticommuting Pauli operators. For a real vector $\mathbf a$, consider $A=\sum_i a_iP_i$. Then $A^2=\sum_i a_i^2\mathbf I$. Since the variance is nonnegative, we obtain
\begin{equation}
    \bigg(\sum_i a_i\mu_{P_i}\bigg)^2\le \sum_ia_i^2.
    \label{eq:anticommuting-family}
\end{equation}
Choosing $a_i=\mu_{P_i}$ gives $(\sum_i \mu_{P_i}^2)^2\le\sum_i \mu_{P_i}^2$ and hence $\sum_i \mu_{P_i}^2\le 1$. 

\noindent{\it Constraints among commuting observables}\\
Stabilizer relations propagate information among commuting products~\cite{Gottesman1997}.
Let $P_1,\ldots,P_m$ be mutually commuting Pauli operators, let $\mathcal K=\{1,\ldots,m\}$, and define
\begin{equation}
    P_{\mathcal K}=\prod_{i=1}^{m}P_i .
    \label{eq:commuting-product}
\end{equation}
Since the $P_i$ commute, $P_{\mathcal K}$ is Hermitian. For an arbitrary
sign vector $\mathbf s=(s_1,\ldots,s_m)\in\{\pm1\}^m$, consider the operator
\begin{equation}
    B(\mathbf s)
    =(m-1)\mathbf I
    +\left(\prod_{i=1}^{m}s_i\right)P_{\mathcal K}
    -\sum_{i=1}^{m}s_iP_i .
\end{equation}
The commuting observables $P_i$ admit a common eigenbasis. On a joint eigenstate,
let $p_i\in\{\pm1\}$ be the eigenvalue of $P_i$. The corresponding eigenvalue
of $B(\mathbf s)$ is
\begin{equation}
    b_{\mathbf s}(\mathbf p)
    =
    m-1
    +\prod_{i=1}^{m}(s_i p_i)
    -\sum_{i=1}^{m}s_i p_i .
\end{equation}
If $\prod_i s_i p_i=+1$, then
\begin{equation}
    b_{\mathbf s}(\mathbf p)
    =
    m-\sum_{i=1}^{m}s_i p_i
    \ge 0 ,
\end{equation}
whereas, if $\prod_i s_i p_i=-1$,
\begin{equation}
    b_{\mathbf s}(\mathbf p)
    =
    m-2-\sum_{i=1}^{m}s_i p_i
    \ge 0 .
\end{equation}
Hence $B(\mathbf s)\succeq0$ for every sign vector $\mathbf s$. Since
$\rho\succeq0$, it follows that
\begin{equation}
    \operatorname{Tr}(\rho B(\mathbf s))\ge0 ,
\end{equation}
and therefore
\begin{equation}
    \sum_{i=1}^{m}s_i\mu_{P_i}
    -
    \left(\prod_{i=1}^{m}s_i\right)\mu_{P_{\mathcal K}}
    \le m-1 ,
    \qquad
    \forall\,\mathbf s\in\{\pm1\}^m .
    \label{eq:commuting-family}
\end{equation}
Equation~\eqref{eq:commuting-family} gives a family of lower and upper
bounds for the expectation value of the product. In particular,
\begin{align}
    \mu_{P_{\mathcal K}}
    &\ge
    \sum_{i=1}^{m}s_i\mu_{P_i}-(m-1),
    &&
    \prod_i s_i=+1 ,
    \label{eq:commutation}
    \\
    \mu_{P_{\mathcal K}}
    &\le
    (m-1)-\sum_{i=1}^{m}s_i\mu_{P_i},
    &&
    \prod_i s_i=-1 .
\end{align}
The tightest interval compatible with these inequalities can be
written as
\begin{eqnarray}
    L_{\mathcal K}(\boldsymbol\mu)
    &=
    \max\left\{
        -1,\,
        \max_{\substack{\mathbf s\in\{\pm1\}^m\\
                        \prod_i s_i=+1}}
        \left[
            \sum_{i=1}^{m}s_i\mu_{P_i}-(m-1)
        \right]
    \right\},
    \\
    U_{\mathcal K}(\boldsymbol\mu)
    &=
    \min\left\{
        1,\,
        \min_{\substack{\mathbf s\in\{\pm1\}^m\\
                        \prod_i s_i=-1}}
        \left[
            (m-1)-\sum_{i=1}^{m}s_i\mu_{P_i}
        \right]
    \right\},
\end{eqnarray}
so that
\begin{equation}
    L_{\mathcal K}(\boldsymbol\mu)
    \le
    \mu_{P_{\mathcal K}}
    \le
    U_{\mathcal K}(\boldsymbol\mu).
\end{equation}

For the application considered in the main text, it is convenient to absorb
the signs of the measured expectations into the observables. Let
\begin{equation}
    H_i=\sigma_i P_i,
    \qquad
    \sigma_i=\operatorname{sgn}(\mu_{P_i}),
    \qquad \operatorname{sgn}(0)=1,
\end{equation}
such that
\begin{equation}
    \mu_{H_i}=|\mu_{P_i}|=1-\epsilon_i,
    \qquad 0\le\epsilon_i\le1 ,
\end{equation}
and define
\begin{equation}
    H_{\mathcal K}=\prod_{i\in\mathcal K}H_i,
    \qquad
    \epsilon_{\Sigma}=\sum_{i\in\mathcal K}\epsilon_i,
    \qquad
    \epsilon_{\max}=\max_{i\in\mathcal K}\epsilon_i .
\end{equation}
Choosing all signs positive in Eq.~\eqref{eq:commuting-family} gives
\begin{equation}
    \mu_{H_{\mathcal K}}
    \ge
    1-\epsilon_{\Sigma}.
\end{equation}
Including the trivial physical bound $\mu_{H_{\mathcal K}}\ge-1$ yields
\begin{equation}
    \mu_{H_{\mathcal K}}
    \ge
    \max\{-1,1-\epsilon_{\Sigma}\}.
\end{equation}

For the upper bound, the tightest sign configuration with
$\prod_i s_i=-1$ is obtained by flipping the observable with the smallest
$|\mu_{H_i}|$, equivalently the largest $\epsilon_i$. This gives
\begin{equation}
    \mu_{H_{\mathcal K}}
    \le
    1+\epsilon_{\Sigma}-2\epsilon_{\max},
\end{equation}
and after imposing $\mu_{H_{\mathcal K}}\le1$,
\begin{equation}
    \max\{-1,1-\epsilon_{\Sigma}\}
    \le
    \mu_{H_{\mathcal K}}
    \le
    \min\{1,1+\epsilon_{\Sigma}-2\epsilon_{\max}\}.
    \label{eq:commuting-compact-bound}
\end{equation}

The width of the resulting interval is
\begin{equation}
    W_{\mathcal K}
    =
    \min\{2,\epsilon_{\Sigma}\}
    -
    \max\{0,2\epsilon_{\max}-\epsilon_{\Sigma}\}.
    \label{eq:commuting-width}
\end{equation}
The interval therefore collapses to a single point if and only if all but
at most one of the $\epsilon_i$ vanishes. In particular, if
$\epsilon_j$ is the only nonzero deviation, then the commuting product is
completely determined by the remaining saturated observables and by
$\mu_{H_j}$.

\noindent{\it Proof of the Result (Part 1)}\\
For independent outcomes in the interval $[-1,1]$, the Hoeffding's inequality gives
\begin{equation}
    \text{Pr}(|\hat \mu_{G_i}-\mu_{G_i}|>a_i)\le 2\exp (-Ma_i^2/2)= \delta/r
\end{equation}
We apply the union bound on the events $\{|\hat \mu_{G_i}-\mu_{G_i}|>a_i\}$ and we obtain
\begin{equation}
    \text{Pr}\bigg(\bigcup_{i=1}^r\{|\hat \mu_{G_i}-\mu_{G_i}|>a_i\}\bigg)\le \sum_{i=1}^r\text{Pr}(|\hat \mu_{G_i}-\mu_{G_i}|>a_i)\le \delta,
\end{equation}
that guarantees an event of probability at least $1-\delta$ on which all observable intervals hold. Taking $s_i=\text{sgn}(\hat \mu_{G_i})$ and $H_i=s_iG_i$, on the event of probability at least $1-\delta$, we have simultaneously for all $i=1,\dots,r$ the inequality $|s_i\hat \mu_{G_i}-s_i\mu_{G_i}|\le a_i$, that equals $||\hat \mu_{G_i}|-\mu_{H_i}|\le a_i$ and gives
\begin{equation}
    \mu_{H_i}\ge |\hat \mu_{G_i}|-a_i.
\end{equation}
Combining it with the physical bound $\mu_{H_i}\ge -1$ we obtain
\begin{equation}
    \mu_{H_i}\ge \max\{-1,|\hat \mu_{g_i}|-a_i\}=1-e_i,
\end{equation}
where $e_i=\min\{2,1-|\hat \mu_{G_i}|+a_i\}$.\\
\noindent{\it Proof of the Result (Part 2)}\\
For the anticommuting pair $\{Q,G_i\}=0$ it holds
\begin{equation}
    \mu_Q^2+\mu_{G_i}^2\le 1,
\end{equation}
so $|\mu_Q|\le \sqrt{1-\mu_{G_i}^2}$. Applying the result of the previous point, with a probability at least $1-\delta$,
\begin{equation}
    |\mu_Q|\le \sqrt{1-(|\hat \mu_{G_i}|-a_i)_+^2}=\sqrt{1-L_i^2},
\end{equation}
where $(x)_+=\max\{0,x\}$. The tightest interval is
\begin{equation}
    |\mu_Q|\le \min_{i:\{Q,G_i\}=0}\sqrt{1-L_i^2}.
\end{equation}\\
\noindent{\it Proof of the Result (Part 3)}\\
Consider Equation \ref{eq:commutation}, and apply it to $P_i=H_i$, with $m=|K|$. This gives
\begin{align}
    \mu_{H_K}&\ge \sum_{i\in K}\mu_{H_i}-(|K|-1)\nonumber\\
    &=1-\sum_{i\in K}(1-\mu_{H_i})\nonumber\\
    &\ge 1-\sum_{i\in K}e_i.
\end{align}
Introducing the physical bound $\mu_{H_K}\ge -1$, we have, with probability at least $1-\delta$,
\begin{equation}
    \mu_{H_K}\ge\max \bigg\{-1, 1-\sum_{i\in K}e_i\bigg\}
\end{equation}

\clearpage

\bibliography{references} 
\bibliographystyle{sciencemag}

%
%
%
%
%
%


\section*{Acknowledgments}
 We acknowledge use of OpenAI GPT-6.1 Sol for assistance with coding, and crosschecking of literature and calculations.
\paragraph*{Funding:}
M. T. is supported by an Honors School Studentship of Politecnico di Torino, funded by the grant number 54$\_$RBA23GD01.
\paragraph*{Author contributions:}
M. T. carried out the calculations, designed the algorithms and ran the numerical simulations. M. T. and D. G. formulated the research problem, conceived the theoretical solution, and wrote the manuscript.
\paragraph*{Competing interests:}
There are no competing interests to declare.
\paragraph*{Data, code and materials availability:}

Source codes for the implementation of the proposed algorithms are provided at this GitHub link:  https://github.com/teddematteo/corr-inf-quant-learning.



\newpage


\renewcommand{\thefigure}{S\arabic{figure}}
\renewcommand{\thetable}{S\arabic{table}}
\renewcommand{\theequation}{S\arabic{equation}}
\renewcommand{\thepage}{S\arabic{page}}
\setcounter{figure}{0}
\setcounter{table}{0}
\setcounter{equation}{0}
\setcounter{page}{1} 

\begin{center}
\section*{Supplementary Materials for\\ \scititle}
 
Matteo Tedde$^{\dagger}$,
Davide Girolami$^{\ast\dagger}$,\\
\small$^\ast$Corresponding author. Email: davegirolami@gmail.com
\end{center}

\subsection*{\label{sec:prot}Measurement records and adaptive weighting}
Define a local Pauli setting $B\in \{X,Y,Z\}^{\otimes n}$. We say that $B \succeq P$ ($P$ is covered by $B$) when $B_j=P_j$ for every $j\in \operatorname{supp}(P)$. Define $\mathbf 1[B_t \succeq P]$ as the indicator function that equals 1 when $P$ is covered by the setting at time $t$, and let $p_t(P)\in \{-1,1\}$ be the corresponding parity at $t$. The uniform classical shadows procedure uses
\begin{equation}
    q^{\text{unif-CS}}(B)=3^{-n},\qquad \pi^{\text{unif-CS}}(P)=\sum_{B\succeq P}q^{\text{unif-CS}}(B)=3^{-\operatorname{wt}(P)}.
    \label{eq:cs-distribution}
\end{equation}
Our ConEnh procedure instead uses a conditional distribution $\tilde q_t^{\text{ConEnh}}$ that is measurable with respect to the past $\mathcal F_{t-1}$. Specifically, let $q_t(B)$ denote the conditional probability of sampling a measurement basis $B$ at time $t$, given the measurement history prior to $t$, and let $p_t(P)\in\{-1,+1\}$ denote the measured parity associated with the Pauli string $P$ at time $t$. We then define
\begin{equation}
    \pi_t(P)=\sum_{B\succeq P}q_t(B),\qquad Z_t(P)=\frac{\mathbf 1[B_t\succeq P]p_t(P)}{\pi_t(P)}
\end{equation}
If $\pi_t(P)>0$ $\forall P\in \mathcal P_n$, then it is guaranteed
\begin{equation}
    \mathbb E[Z_t(P)|\mathcal F_{t-1}]=\mu_P.
\end{equation}
The filtration $\mathcal F_{t-1}$ includes previous records and any pilot, pool, candidate randomization realized for the design up to time $t-1$. Raw moments under the state sampling models are unbiased. However, optional Pauli clipping and final property projections that regularize the estimates introduce bias to the final estimator.\\
To ensure full support, we add a uniform exploration component of weight $\eta \in (0,1]$,
\begin{equation}
    q^{\text{ConEnh}}_t(B)=(1-\eta)\tilde q_t^{\text{ConEnh}}(B)+\eta 3^{-n},
    \label{eq:bi-mixture}
\end{equation}
which gives
\begin{equation}
    \pi^{\text{ConEnh}}_t(P)=\sum_{B\succeq P} q^{\text{ConEnh}}_t(B)=(1-\eta)\sum_{B\succeq P}\tilde q_t^{\text{ConEnh}}(B)+\eta 3^{-\operatorname{wt}(P)}.
    \label{eq:bi-inclusion}
\end{equation}
Since $q_t$ is constructed before the observation at time $t$,
\begin{equation}
\mu_P=    \mathbb E[Z_t(P)|\,\mathcal F_{t-1}],\qquad Z_t(P)=\frac{\mathbf 1[B_t \succeq P]p_t(P)}{\pi_t(P)}.
    \label{eq:adaptive-unbiased}
\end{equation}
This is the backbone for all our CS and ConEnh estimators. For $s<t$, iterated conditioning gives $\mathbb E[Z_s(P)Z_t(P)]=\mu_P^2$. As a result, the estimator
\begin{equation}
    U_N(P)=\binom{N}{2}^{-1}\sum_{s<t}Z_s(P)Z_t(P),
    \label{eq:u-estimator}
\end{equation}
where $N$ is the number of shots, is unbiased for $\mu_P^2$.

\paragraph*{Stabilizer-informed classical shadows (CS).}
CS uses the same reference generators as ConEnh, without access to
rotation sites, angles or exact expectation values. For a subsystem $A$, we retain generators whose support intersects $A$, keeping their full physical support. Let $B_1,\ldots,B_L$ be the resulting distinct bases. They receive equal weights $1/L$,
fixed before acquisition. Each shot samples
\begin{equation}
    q^{\mathrm{CS}}(B)
=\frac{1-\eta}{L}\sum_{\ell=1}^{L}\mathbf{1}[B=B_\ell]
+\eta 3^{-n}.
\end{equation}
We use the same exploration weight $\eta$ as ConEnh. The exact inclusion probability of a Pauli observable is
\begin{equation}
\pi^{\mathrm{CS}}(P)
=\frac{1-\eta}{L}\sum_{\ell=1}^{L}\mathbf{1}[B_\ell\succeq P]
+\eta 3^{-\mathrm{wt}(P)}.
\end{equation}

\subsection*{\label{sec:prot}Direct entanglement estimation}
\paragraph{Uniform classical shadows estimator.}
For completeness, we recall the conventional uniform local-Pauli classical shadows estimator, which evaluates the full Pauli sum through a factorized expression. This differs from the CS protocol used in our numerical comparisons. Indeed, the Pauli bases are chosen independently and uniformly on each qubit, which allows a factorization of the estimator in Equation \ref{eq:u-estimator}. For the shots at time $s,t$, let $b_{s,j}$ and $b_{t,j}$ be their local outcomes in the qubit $j$. Then
\begin{align}
    K_{st}^{\text{unif-CS}}&=2^{-|A|}\sum_{P\in \mathcal P_{|A|}}Z_s(P)Z_t(P)\\
    &=\prod_{j\in A}\frac{1+9b_{s,j}b_{t,j}\mathbf 1[B_{s,j}=B_{t,j}]}{2}.
\end{align}
It follows
\begin{equation}
    \widehat{\text{Tr}\left(\rho_A^2\right)^\text{unif-CS}}=\binom{N}{2}^{-1}\sum_{s<t}K_{st}^\text{unif-CS}.
\end{equation}
This approach avoids the enumeration of $4^{|A|}$ Pauli coordinates. In current code the storage is linear in $N$, while pair processing is quadratic in $N$. The estimator for the bipartite  2-Rényi entanglement is then computed as
\begin{equation}
    \widehat{\mathcal E_2 (\psi_{AB})^\text{unif-CS}}=-\log_2 \widehat{\text{Tr}(\rho^2_A)}.
\end{equation}

\paragraph{Importance-sampled estimator.} Previous factorization cannot be applied to ConEnh and CS, due to the nature of $q_t^{\text{ConEnh}}(B)$ and $q^{\text{CS}}(B)$. For this reason, we design an importance-sampled estimator, where only a subset of Pauli strings is sampled to compute the sum, according to a given distribution. Let $\mathcal S_A$ be the set of operators obtained as a product of the known commuting Pauli observables $H_i$ and having support in A
\begin{equation}
    \mathcal{S}_A =
\left\{
\prod_{i=1}^{r} H_i^{x_i}
\;\middle|\;
x_i \in \{0,1\}, \quad
\operatorname{supp}\left(
\prod_{i=1}^{r} H_i^{x_i}
\right) \subseteq A
\right\}
\setminus \{I\}.
\end{equation}
$\mathcal S_A$ has cardinality $|\mathcal S_A|=2^{d_A}-1$, where $d_A$ is the number of independent generators of $\mathcal G_A$. We define the following distribution
\begin{equation}
    \nu_A(P)=\frac{\gamma}{4^{|A|}-1}+(1-\gamma)\frac{\mathbf 1[P\in \mathcal S_A]}{2^{d_A}-1}.
\end{equation}
This distribution combines the contributions of Pauli observables internal to $\mathcal S_A$ and Pauli observables outside this set, which are still relevant in the sum of the purity due to the applied rotations. We proceed by sampling $S$ operators $P_1,\dots,P_S\overset{\mathrm{i.i.d.}}{\sim} \nu_A(P)$ and building the set $\mathcal U$ formed by distinct observables. If an operator $P$ shows $c_P$ times in the sampling, it is assigned the weight
\begin{equation}
    W_P=\frac{c_P}{S\nu_A(P)}.
\end{equation}
These weights allow us to build the unbiased estimator
\begin{equation}
    \widehat{\text{Tr}\left(\rho_A^2\right)}=2^{-a}\bigg[1+\sum_{P\in \mathcal U}W_PU_N(P)\bigg].
\end{equation}
So instead of summing all the Pauli observables having support in $A$, we estimate the sum using sampled Pauli observables. This estimator is unbiased, but it is not guaranteed to have low variance.

\paragraph{Clipping and property projection.} For each generator $H_i$ (or Pauli of the known commuting set), ConEnh records the observation count $n_i$ and its parity sum. When $n_i>0$, its empirical mean is $\hat \mu_{H_i}$ and its empirical ``deficiency" is $\hat e_i=1-|\hat \mu_{H_i}|$. If $n_i=0$, $\hat \mu_{H_i}$ is set to 0. Hence, according to Equation \ref{eq:anticommutingbound}, for each Pauli $P$ that anticommutes with at least one $H_i$ we can define the heuristic upper bound for $\mu_P^2$:
\begin{equation}
    v_P=\min \{1,1-\hat \mu_{H_i}^2:\{P,H_i\}=0\}.
\end{equation}
This clipping operation  is motivated by the anticommutation constraints in Eq.~(3), with empirical generator estimates replacing their unknown expectation values.
ConEnh uses
\begin{equation}
    \widehat{\text{Tr}\left(\rho_A^2\right)_\text{Clipped}}=2^{-|A|}\bigg[1+\sum _{P\in \mathcal U}W_P \text{clip}(U_N(P),0,v_P)\bigg],
\end{equation}
where
\begin{equation}
    \text{clip}(x,a,b)=\begin{cases}
        a\qquad \text{if}\quad x<a,\\
        x\qquad \text{if}\quad a\le x\le b,\\
        b\qquad \text{if}\quad x>b.
    \end{cases}
\end{equation}
Clearly the accuracy of the estimator $\hat \mu_{H_i}$, and consequently of the clipping cap $v_P$, improves as the number of measurement shots increases. Furthermore, this empirical cap introduces bias to the estimator of the local purity, but in many simulations we observe that it can effectively suppress outliers. \\
In addition, both CS and ConEnh methods use the same physical projection before the logarithm:
\begin{equation}
\widetilde{\text{Tr}\left(\rho_A^2\right)}=\text{clip}\bigg(\widehat{\text{Tr}\left(\rho_A^2\right)},2^{-|A|},1\bigg).
\end{equation}
Clearly, this global projection and logarithm introduce further bias.

\subsection*{ConEnh measurement design for entanglement}
\paragraph{Candidate construction.}\label{par:candidateconstr} Consider the known generators $H_i,\dots,H_r$ and select the ones intersecting the subsystem $A$
\begin{equation}
    \text{supp}(H_i)\cap A\ne \emptyset.
\end{equation}
These generators provide information about the constraints relevant to the subsystem. ConEnh builds a set of \textit{cliques}
\begin{equation}
    \mathcal C=\{C_1,C_2,\dots\},
\end{equation}
where each clique contains generators that can be measured by the same setting $B$
\begin{equation}
    C \subseteq \{H_1,\ldots,H_r\},
\qquad
\forall i,k \in C,\ \forall j,
\quad
H_{i,j}=H_{k,j}
\;\lor\;
H_{i,j}=I
\;\lor\;
H_{k,j}=I.
\end{equation}
Here $j\in\{1,\dots,n\}$ identifies a qubit and $H_{i,j}$ is the local Pauli operator of the generator $H_i$. For each generator $H_i$ define
\begin{equation}
    u_i=\frac{1}{\sqrt{1+n_i}},
\end{equation}
where $n_i$ denotes how many times the generator has been observed. The quantity
\begin{equation}
    g(C)=\sum_{i\in C}u_i
\end{equation}
decreases as we measure generators of the clique $C$. The algorithm then extracts $J$ candidate bases from the subsystem $A$
\begin{equation}
    B_A^{(j)}\in \{X,Y,Z\}^{\otimes |A|},\qquad j=1,\dots,J
\end{equation}
and for each basis it selects the clique that maximizes $g(C)$
\begin{equation}
    C_j=\arg\max_{C\in \mathcal C:C\ni B_A^{(j)}} g(C).
\end{equation}
Where $C\ni B_A^{(j)}$ means that all the generators of the clique $C$ are compatible with the basis $B_A^{(j)}$ on the qubits of $A$. This allows to exploit a measurement not just to measure Pauli observables having support in $A$, but also generators that act outside $A$. The candidate's generator score for $B_A^{(j)}$ is therefore $g_j=g(C_j)$.

\paragraph{Scores and sampling probabilities.} Let $m_P$ count earlier settings covering $P\in \mathcal U$, let $a_i$ count earlier full-support observables on $A$ anticommuting with $H_i$ (to be not confused with Hoeffding radius of previous sections), and let $l_{j,b}$ count uses of direction $b$ on the qubit $j\in A$. The residual rewards are
\begin{equation}
    r_P=\frac{W_P\max \{v_p,0.01\}}{1+m_P},\qquad s_i=\frac{\hat e_i}{1+a_i}.
\end{equation}
Define $R_C=\sum_Pr_P$, $R_A=\sum_is_i$, and
\begin{equation}
    \alpha=\frac{\beta R_C}{\beta R_C+R_A}.
\end{equation}
For candidate $B_A$, let $P(B_A)$ be its full-support Pauli on $A$. We consider the following components
\begin{align}
    C(B_A)&=R_C^{-1}\sum_{P\in \mathcal U:B_A\succeq P}r_P,\\
    A(B_A)&=R_A^{-1}\sum_{i:\{P(B_A),H_i\}=0}s_i\\
    L(B_A)&=\frac{1}{|A|}\sum_{j\in A}\frac{1}{1+l_{j,B_j}}.
\end{align}
The coverage and total scores are
\begin{align}
    S_{\text{cov}}(B_A)&=\alpha C(B_A)+(1-\alpha)A(B_A)+\lambda_{loc}L(B_A),\\
    S_j&=\lambda_p S_\text{cov}(B_A^{(j)})+\lambda_G\frac{g_j}{\max_l g_l}.
\end{align}
To summarize, $C(B_A)$ rewards Pauli observables relevant to the target property, $A(B_A)$ rewards Pauli observables that anticommute with generators, $L(B_A)$ rewards Pauli observables that cover directions poorly explored, and $g(C_j)$ rewards poorly measured generators. For a given candidate basis $B_A^{(j)}$ we build the softmax
\begin{equation}
    p_j=\frac{\exp[(S_j-S_{\text{max}})/\tau]}{\sum_l\exp[(S_l-S_{\text{max}})/\tau]},
\end{equation}
and define the probability distribution on $P$ at time $t$ as
\begin{equation}
    \pi_t^{\text{ConEnh}}(P)=(1-\eta)\sum_{j:B_A^{(j)}\succeq P}p_j+\eta 3^{-\text{wt}(P)}.
\end{equation}

\subsection*{Neural AGHDO reconstruction of magic}
The magic comparison uses the same physical neural model, training objective, architecture, and property evaluation for CS-AGHDO and ConEnh-AGHDO. Each strategy trains its own network on its own measurement records.

\paragraph{Autoregressive Gram-Hadamard density operator.}
AGHDO parametrizes density-matrix elements by
\begin{equation}
    \rho_\phi(\boldsymbol{\sigma},\boldsymbol{\eta})=\prod_{h=0}^{n-1}\sum_{r=1}^R \psi _{h,\phi}(\sigma_h,\boldsymbol{\sigma}_{<h},r)\psi _{h,\phi}(\eta_h,\boldsymbol{\eta}_{<h},r)^*.
\end{equation}
Each factor is a Gram kernel, so their product is positive semidefinite. A further conditional normalization and sequential summation of the diagonal give unit trace. The rank of the global density operator can grow up to $\min\{2^n,R^n\}$.  In the current magic benchmark, we exploit the known product-state structure of the target preparations in a fixed Clifford reference frame. This allows us to restrict the general autoregressive ansatz to the following parametrization of the pure states we are considering:
\begin{equation}
    \rho_\phi=C\bigg(\bigotimes_{i=1}^n\varrho_{\phi,i}\bigg)C^\dagger,\qquad \varrho_{\phi,i}(s,t)=\sum_{a=1}^R\psi_{\phi,i}(s,a)\psi^*_{\phi,i}(t,a),
    \label{eq:aghdo-density2}
\end{equation}
where $C$ is the known Clifford transformation. This is computationally convenient, since it allows to build $n$ $2 \times 2$ density matrices, instead of parametrizing a $2^n\times 2^n$ matrix. Indeed, each $\varrho_{\phi,i}$ is the density matrix associated to the logic qubit $i$ and is parametrized by a neural network with two hidden layers of widths 8 and 4, producing 16 complex outputs.

\paragraph{Physical focus observables.}
In the known Clifford reference frame, the target states have a product structure, so each logical qubit $i$ is completely characterized by its three Bloch coordinates
$(x_i,y_i,z_i)$.
This structure is particularly useful for estimating the stabilizer Rényi-2 entropy. Indeed, since the fourth Pauli moment is invariant under Clifford
transformations and factorizes over tensor-product states, we have

\begin{equation}
\mathcal A_4(\rho)=
\prod_{i=0}^{n-1}
\frac{1+x_i^4+y_i^4+z_i^4}{2},
\label{eq:magicaghdo}
\end{equation}
and therefore

\begin{equation}
\mathcal M_2(\rho)=
-\sum_{i=0}^{n-1}
\log_2\left(
\frac{1+x_i^4+y_i^4+z_i^4}{2}
\right).
\end{equation}
Consequently, within the adopted product-frame model, the magic can be evaluated from the three Bloch coordinates of each logical qubit, without enumerating all $4^n$ Pauli operators.
However, measurements are performed in the physical reference frame. We therefore define the physical focus observables as the Clifford images of the logical single-qubit Pauli operators:

\begin{equation}
Q_{i,\alpha}
=
C\sigma_i^\alpha C^\dagger,
\qquad
i=0,\ldots,n-1,
\quad
\alpha\in\{X,Y,Z\}.
\end{equation}
Thus, for $n=100$, only $3n=300$ focus observables are required to characterize the local density matrices within the adopted model and evaluate the global magic.
For the cluster and GHZ state families, these observables take the following explicit forms:

\begin{center}
\renewcommand{\arraystretch}{1.5}
\begin{tabular}{c|c|c|c}
\hline
Family & $Q_{i,X}$ & $Q_{i,Y}$ & $Q_{i,Z}$ \\
\hline
Cluster
&
$X_i\prod_{j\in N(i)}Z_j$
&
$Y_i\prod_{j\in N(i)}Z_j$
&
$Z_i$
\\
\hline
GHZ hub
&
$X_0\prod_{j>0}X_j$
&
$Y_0\prod_{j>0}X_j$
&
$Z_0$
\\
\hline
GHZ leaf
&
$Z_0Z_i$
&
$-Z_0Y_i$
&
$X_i$
\\
\hline
\end{tabular}
\end{center}
The longitudinal observable $Q_{i,X}$ coincides with the reference generator $H_i$. The transverse observables $Q_{i,Y}$ and $Q_{i,Z}$ provide additional information about the corresponding logical qubit.

\paragraph{ConEnh-AGHDO acquisition policy.} The selector used in the current benchmark is different from the subsystem selector in the previous section, despite being based on the same principles. It tracks generator means $\hat \mu_{H_i}$, counts $n_i$, and coverage counts $m_{i,\alpha}$ of the observable $Q_{i,\alpha}$. For unobserved generators set $\hat \mu_{H_i}=0$ and define
\begin{equation}
    \hat e_i=\min \{1,1-|\hat \mu_{H_i}|+(1+n_i)^{-1/2}\},\qquad w_i=[\hat e_i(2-\hat e_i)]^2,\qquad r_{i,\alpha}=\frac{w_i}{1+m_{i,\alpha}}.
\end{equation}
The idea is that for $Q_{i,Y}$ and  $Q_{i,Z}$ it holds
\begin{equation}
    \mu_{Q_{i,\alpha}}^2+\mu_{H_i}^2\le 1.
\end{equation}
Writing $|\mu_{H_i}|=1-e_i$ we obtain
\begin{equation}
    (1-\mu_{H_i}^2)^2=(1-(1-e_i)^2)^2=(e_i(2-e_i))^2,
\end{equation}
which means that when $H_i$ is nearly deterministic the contributions of $Y$ and $Z$ are limited by anticommutation, decreasing their score. For the candidate basis $B^{(j)}$
\begin{equation}
    S_j=\frac{\sum_{i,\alpha:B^{(j)}\succeq Q_{i,\alpha}}r_{i,\alpha}}{\sum_{i,\alpha}r_{i,\alpha}},\qquad p_j=\frac{e^{(S_j-S_{\max})/\tau}}{\sum_l e^{(S_l-S_{\max})/\tau}},
\end{equation}
which can be used to build the basis distribution described in previous sections. 

\paragraph{Likelihood training and model evaluation.}
Let $N_{i,\alpha}^+$ and $N_{i,\alpha}^-$ denote the count of the logical outcomes $\pm 1$. Define $\xi_{i,\alpha}=\text{Tr}(\rho_{i,\phi}\sigma^\alpha)$. The loss is
\begin{equation}
    \mathcal L(\phi)=-\frac{1}{\sum_{i,\alpha,b}N_{i,\alpha}^b}\sum_{i,\alpha}\sum_{b=\pm 1} N_{i,\alpha}^b \log\bigg[\max\bigg\{10^{-12},\frac{1+b\xi_{i,\alpha}}{2}\bigg\}\bigg].
\end{equation}
At every checkpoint, Adam optimizer takes 12 steps at learning rate 0.005, using all accumulated counts and retaining the parameters and optimizer state from the previous checkpoint. Adaptive basis probabilities do not depend on the fitted model parameters, so their factors in the conditional data likelihood do not enter this optimization. After the optimization step, we can compute the estimators as follows
\begin{equation}
    \hat x_i=\hat \mu_{Q_{i,X}}=\text{Tr}(\rho_{i,\phi^\star}X),\qquad
    \hat y_i=\hat \mu_{Q_{i,Y}}=\text{Tr}(\rho_{i,\phi^\star}Y),\qquad
    \hat z_i=\hat \mu_{Q_{i,Z}}=\text{Tr}(\rho_{i,\phi^\star}Z),
\end{equation}
and apply Equation \ref{eq:magicaghdo}.

\paragraph{Model error decomposition.}
Restricting the search to a parametrized physical family incurs a model error. Let $\mathcal F_{\mathrm{AGHDO}}=\{\rho_{\phi}\}$ be the family modeled by AGHDO, and let $\rho_{\phi^\star}$ be the optimum within the measurement model. For a target function $f$,
\begin{equation}
    | f (\hat \rho)- f(\rho)|\le|f(\hat \rho)-f(\rho_{\boldsymbol\phi^\star})|+|f(\rho_{\boldsymbol\phi^\star})-f(\rho)|.
    \label{eq:aghdo-error-decomposition}
\end{equation}
The first term contains the estimation and optimization error of the fitted model, while the second term is the model error, which persists whenever the true state cannot be represented at the chosen rank. Positivity can reduce variance by eliminating unphysical directions, but it does not remove approximation bias. CS-AGHDO and ConEnh-AGHDO use the same model family, optimizer and hyperparameters, with measurement acquisition as the controlled experimental difference. This does not guarantee identical model error: under model misspecification, the population-likelihood optimum may depend on the acquisition distribution.

\end{document}